\documentclass[a4paper,12pt]{article}
\usepackage[T1]{fontenc}
\pdfoutput=1

\usepackage{jheppub}
\usepackage{macro}
\usepackage{appendix}
\usepackage{comment}
\usepackage{array}
\usepackage{yfonts}
\usepackage{comment}
\usepackage{amsthm}
\usepackage{tikz-cd}
\usepackage{amsmath,amssymb}
\usepackage{tikz}
\usepackage{pgfkeys}
\usepackage{pgfopts}
\usepackage{xcolor}
\usepackage{young}
\usepackage{youngtab}
\usepackage{ytableau}
\usepackage{titlesec}
\usetikzlibrary{decorations.pathreplacing}

\usepackage{tikz}

\titleformat{\paragraph}
{\normalfont\normalsize\bfseries}{\theparagraph}{1em}{}
\titlespacing*{\paragraph}
{0pt}{3.25ex plus 1ex minus .2ex}{1.5ex plus .2ex}

\newtheorem{Theorem}{Theorem}[section]

\newtheorem{Definition}{Definition}[section]
\newtheorem{Lemma}{Lemma}[section]
\newtheorem{Proposition}{Proposition}[section]

\usetikzlibrary{matrix,calc,positioning,decorations.markings,decorations.pathmorphing,decorations.pathreplacing
}%tikzmark,
\usetikzlibrary{arrows,cd,shapes}

\title{\textbf On the solution of the harmonic-divgrad PDE system}

\author[a,b]{Federico Manzoni}

\affiliation[a]{Mathematics and Physics department, Roma Tre, Via della Vasca Navale 84, Rome, Italy}
\affiliation[b]{INFN Roma Tre Section, Physics department, Via della Vasca Navale 84, Rome, Italy}
\emailAdd{federico.manzoni@uniroma3.it, ORCID ID: 0000-0002-9979-6154}

\abstract{We study a particular system of partial differential equations in which the Laplace--Beltrami, the  divergence and the gradient operators of the unknown functions appear (harmonic-divgrad system). The analysis of this particular system is motivated by its occurrence in the study of asymptotic symmetries in $p$-form gauge theories and in mixed symmetry tensor gauge theories. Using the Killing--Hopf theorem and leveraging the properties of Riemannian manifolds with constant sectional curvature we establish the conditions under which these equations admit only the trivial solutions proving their trivialization on positive curvature space forms.}

\begin{document}

\maketitle
%\flushbottom
%\newpage

\section{Introduction}

Mathematical modeling of physical systems often involves dealing with differential equations and partial differential equations. Motivated by classical field theory, specifically by gauge theories where the gauge field is a $p$-form or a mixed symmetry tensor \cite{Henneaux:1986ht,Curtright:1980un,Curtright:1980yj,Curtright:1980yk} we study a particular system of partial differential equations which appears in such theories. Additionally, motivated also by the study of asymptotic symmetries in gauge theories and/or gravity (a non-exhaustive list of classic, recent and reviewing, articles on the subject is \cite{Sachs:1962wk,Bondi:1962px,Anderson1992IntroductionTT,Barnich_2010,strominger2018lectures,Afshar:2018apx,Manzoni:2021dij,Ciambelli:2022vot,Manzoni:2025gxw, Manzoni:2024agc,Romoli:2024hlc,Ferreira:2016hee,Flanagan_2020,Compere:2018aar,barnich:2009se,Romoli:2025map,Manzoni:2024tow,manzf,Manzoni:2025zmi,Manzoni:2026skv}) we are mainly interested in considering the $(D-2)$-dimensional sphere, known as the celestial sphere, as manifold. By extension, we consider functions defined on a maximally symmetric space, mathematically referred to as space form. Physically, these spaces admit the maximal number of Killing vectors while mathematically they are complete Riemannian manifolds with constant sectional curvature.\\
A powerful theorem of Riemannian geometry, the Killing--Hopf theorem \cite{Killing1891,Hopf1926ZumCR,lee2019introduction}, ensures that every space form admits a Riemannian covering given by one of the sphere, the hyperbolic space or the Euclidean space depending on the sign and value of the sectional curvature. Thanks to this theorem we are able to prove that the particular system of second-order partial differential equations, given in Definition \ref{def}, on a positive curvature space form $M$ for a function $f \in C^{\infty}(M)$ and a 1-form $h \in \Omega^1(M)$ admits only the trivial solution away from an explicit excluded resonant set for the parameters ($k_1$,$k_2$), as defined in Lemma \ref{Lem1} and determined by the Laplacian spectrum. The proof strategy consists in reducing the second-order coupled PDEs system to a single scalar fourth-order operator $L$ acting on $f-\nabla^ih_i$ (Proposition \ref{Prop1}), exclude the resonant set for the parameters ($k_1$,$k_2$) in order to ensure the kernel triviality for $L$ (Lemma \ref{Lem1}), prove the trivialization on $S^d$ using energetic arguments (Theorem \ref{trivsys}) and finally use Killing--Hopf theorem to conclude the trivialization on a positive sectional curvature space form (Theorem \ref{trivsysext}).

\section{Space form and Killing–Hopf theorem}
Let us start with the definition of space form.
\begin{Definition}[Space form]
    A space form is a complete Riemannian manifold $(M,\boldsymbol{g})$ of constant sectional curvature $K$.
\end{Definition}
Three fundamental examples are the Euclidean $d$-dimensional space $(E^d,\boldsymbol{\delta})$, the $d$-dimensional sphere $(S^d,\boldsymbol{\gamma}_S)$, and the $d$-dimensional hyperbolic space $(H^d,\boldsymbol{\gamma}_H)$, although a space form does not need to be simply connected. The fundamental theorem we are interested in is the following one. 
\begin{Theorem}[Killing--Hopf]\label{tkh}
    Let $(M,\boldsymbol{g})$ be a complete $d$-dimensional Riemannian manifold of constant sectional curvature $K$, then the Riemannian universal cover $(\Tilde{M},\tilde{\boldsymbol{g}})$ of $(M,\boldsymbol{g})$ is 
    \begin{itemize}
        \item the $d$-dimensional sphere if $K>0$;
        \item the Euclidean $d$-dimensional space if $K=0$;
        \item the $d$-dimensional hyperbolic space if $K<0$.
    \end{itemize}
\end{Theorem}
\begin{proof}
    The reader interested in the proof can find it, for example, in \cite{lee2019introduction}. The original references are \cite{Killing1891,Hopf1926ZumCR}.
\end{proof}
Diagrammatically, the Killing-Hopf theorem can be viewed as 
\begin{equation}  
\tikzset{every picture/.style={line width=0.75pt}}
\begin{tikzpicture}[x=0.75pt,y=0.75pt,yscale=-1,xscale=1]
%uncomment if require: \path (0,330); %set diagram left start at 0, and has height of 330

%Straight Lines [id:da8313382710050826] 
\draw    (189,101) -- (189,191.5) ;
\draw [shift={(189,193.5)}, rotate = 270] [color={rgb, 255:red, 0; green, 0; blue, 0 }  ][line width=0.75]    (10.93,-3.29) .. controls (6.95,-1.4) and (3.31,-0.3) .. (0,0) .. controls (3.31,0.3) and (6.95,1.4) .. (10.93,3.29)   ;

% Text Node
\draw (180,73.4) node [anchor=north west][inner sep=0.75pt]    {$\tilde{M}$};
% Text Node
\draw (180,198.4) node [anchor=north west][inner sep=0.75pt]    {$M$};
% Text Node
\draw (164,140.4) node [anchor=north west][inner sep=0.75pt]    {$\pi $};
\end{tikzpicture}
\end{equation}
where the pullback of the covering map is such that 
\begin{equation}
    \tilde{\boldsymbol{g}}=\pi^*\boldsymbol{g}.
\end{equation}
Therefore, according to Killing--Hopf theorem, any complete Riemannian manifold of constant sectional curvature is the quotient of one of the three canonical examples above by a group, which is a subgroup of the corresponding isometry group, that acts freely and properly discontinuously. As a result, only very few smooth manifolds can admit a constant sectional curvature metric; despite this, these spaces are among the most widely used Riemannian manifolds in theoretical physics. A first example are the constant time sections of standard cosmological models \cite{Calcagni:2017sdq}. Such Riemannian manifolds appear also in string theory and holography; indeed, a particular class of Sasaki-Einstein manifolds used to extend the standard AdS/CFT correspondence, the lens space $L(p,q_1,..,q_n)$, are nothing but that positive constant sectional curvature manifold \cite{Ionicioiu:1998qg,Ogawa:2022fyc,Maldacena_1999,Manzoni:2022htx,Aharony:2003sx,Amariti:2022dui,Antinucci:2021edv}. They are quotient of $S^{2n-1}$ by $\Gamma$ that
is a finite cyclic group acting on the sphere by isometries. 

\section{A result on the harmonic-divgrad PDEs system}
Let  us start with the definition of the harmonic-divgrad PDEs system.
\begin{Definition}[Harmonic-divgrad system]\label{def}
Let $(M,\boldsymbol{g})$ be a smooth $d$-dimensional Riemannian manifold with sectional curvature $K$, endowed with its Levi-Civita connection $\nabla$. Let
\begin{equation}
f \in C^\infty(M) \, , \qquad h \in \Omega^1(M) \, .
\end{equation}
Writing locally
\begin{equation}
h = h_i \, dx^i \, ,
\end{equation}
we define the harmonic-divgrad system as
\begin{subequations}
\begin{alignat}{2}
    &[-k_1+\Delta]f-2\nabla^ih_i=0,
        \label{C1a}\\
    &[-k_2+\Delta]h_i+2\nabla_if=0.
    \label{C1b}
\end{alignat}
\end{subequations}
Equivalently, in intrinsic form, the system can be written as
\begin{subequations}
\begin{alignat}{2}
(-k_1 + \Delta) f + 2 \, \delta h &= 0 \, , \\
(-k_2 + \Delta) h + 2 \, df &= 0 \, ,
\end{alignat}
\end{subequations}
where $\Delta = \nabla^i \nabla_i$ denotes the Laplace-Beltrami operator, $\delta$ is the codifferential with the standard notation $\delta h=-\nabla^ih_i$ and $d$ the differential such that $df=\nabla_if$. Throughout the paper, indices are raised and lowered with the metric $\boldsymbol{g}$ and the Einstein summation convention is understood.
\end{Definition}
This kind of system appears in some physical contexts such as exotic gauge theories: its employment is unavoidable in the study of gauge-for-gauge chains in $p$-forms gauge theories especially in cases of high dimensionality of the null-infinity sphere $S^{D-2}$ where $D$ is the dimension of the Minkowski space-time where the gauge theory is defined. Possible extensions of this PDEs system with more generic unknown could be essential for the explicit computations of the asymptotic symmetries in mixed symmetry tensor gauge theories \cite{manzf,Manzoni:2024tow,Manzoni:2025zmi}.
\\
Our main goal is to show the trivialization of the solution of the harmonic-divgrad PDEs system on a positive sectional
curvature space form. The first step is to show that both $f$ and $\nabla^ih_i$ satisfy the same differential equation on the sphere so that their difference also satisfies it.
\begin{Proposition}[Reduction to a single scalar fourth-order equation]\label{Prop1}
    Let 
    \begin{equation}
f \in C^\infty(S^d) \, , \qquad h \in \Omega^1(S^d) \, ,
\end{equation}
satisfying the harmonic-divgrad system on the sphere $S^d$. Let \begin{equation}
    L:=\Delta^2+(\Bar{k}_2-k_1+4)\Delta-k_1\Bar{k}_2
\end{equation} with $\bar{k}_2:=\left(K(d-1)-k_2\right)$ and $(k_1,k_2) \in \mathbb{R}^+ \times \mathbb{R}^+$, then $f-\nabla^ih_i \in C^{\infty}(S^d)$ satisfies
\begin{equation}
    L(f-\nabla^ih_i)=0.
\end{equation}
\end{Proposition}
\begin{proof}
We start by applying $\nabla^i$ to equation \eqref{C1b}:
\begin{equation}
    2\Delta f-k_2\nabla^ih_i+\nabla^i\nabla^j\nabla_j h_i=0,
    \label{C2}
\end{equation}
where, due to curvature effects, we cannot exchange covariant derivatives; however, we can make use of the Riemann tensor components identity $[\nabla^i,\nabla^j]h_i=\nabla^i\nabla^jh_i-\nabla^j\nabla^ih_i=R^k_{iij}h_k=g^{ks}R_{siij}h_k$. Indeed
\begin{equation}
\begin{aligned}
\nabla^i\nabla^j(\nabla_jh_i)&=\nabla^j\nabla^i(\nabla_jh_i)+R^j_{kij}(\nabla^kh^i)+R^i_{kij}(\nabla^jh^k)=\\
&=R^j_{kij}(\nabla^kh^i)+R^i_{kij}(\nabla^jh^k)+\nabla^j\nabla_j\nabla_ih^i+\nabla^j(R^i_{kij}h^k)=\\
&=\nabla^j\nabla_j\nabla_ih^i-R_{ki}(\nabla^kh^i)+R_{kj}(\nabla^jh^k)+\nabla^j(R_{kj}h^k)=\\
&=\nabla^j\nabla_j\nabla_ih^i+\nabla^j(R_{kj}h^k)
\end{aligned}    
\end{equation}
but the Riemann tensor components for a maximally symmetric space with sectional curvature $K>0$ are given by
\begin{equation}
    R_{\alpha \beta \gamma \delta}=K\left(g_{\alpha \gamma}g_{\beta \delta}-g_{\alpha \delta}g_{\beta \gamma}\right),
\end{equation}
from which, using the metric of the $d$-dimensional sphere $\gamma_{\alpha \beta}$, we get
\begin{equation}
    R_{\beta \delta}=\gamma^{\alpha \gamma} R_{\alpha \beta \gamma \delta}=Kd\gamma_{\beta \delta}-K\delta^{\gamma}_{\delta}\gamma_{\beta \gamma}=K(d-1)\gamma_{\beta \delta}.
\end{equation}
Therefore 
\begin{equation}
\nabla^i\nabla^j(\nabla_jh_i)=\nabla^j\nabla_j\nabla^ih_i+K(d-1)\nabla^ih_i.
\end{equation}
Hence, returning to equation \eqref{C2}, we can write
\begin{equation}
    2\Delta f+\left[\Delta+\left(K(d-1)-k_2\right)\right]\nabla^ih_i=0,
\end{equation}
from which
\begin{equation}
    \Delta f=-\frac{1}{2}[\Delta+\Bar{k}_2]\nabla^ih_i,
    \label{soldf}
\end{equation}
where we defined $\Bar{k}_2:=\left(K(d-1)-k_2\right)$.
Now, inserting the expression of $\Delta f$ from above in equation \eqref{C1a} we get
\begin{equation}
    -\frac{1}{2}[\Delta+\Bar{k}_2]\nabla^ih_i=2\nabla^ih_i+k_1f,
\label{eqf}    
\end{equation}
from which 
\begin{equation}
    f=-\frac{[\Delta+\Bar{k}_2+4]}{2k_1}\nabla^ih_i,
\label{solf}    
\end{equation}
whose reinsertion into the original equation \eqref{C1a} gives 
\begin{subequations}
\begin{alignat}{1}
&-[-k_1+\Delta]\frac{[\Delta+\Bar{k}_2+4]}{2k_1}\nabla^ih_i-2\nabla^ih_i=0
\end{alignat}
\end{subequations}
and, rearranging, we can write 
\begin{equation}
 [\Delta^2+(\Bar{k}_2-k_1+4)\Delta-k_1\Bar{k}_2]\nabla^ih_i=0.
\end{equation}
In similar way, from equation \eqref{C1a}, we can find
\begin{equation}
    \nabla^ih_i=\frac{1}{2}[-k_1+\Delta]f
\end{equation}
and, substituting it in equation \eqref{C2}, we get
\begin{equation}
    2\Delta f=-\frac{1}{2}[\Delta+\Bar{k}_2][-k_1+\Delta]f \quad \Rightarrow \quad [\Delta^2+(\Bar{k}_2-k_1+4)\Delta-k_1\Bar{k}_2]f=0.
\end{equation}
Therefore, both functions $\nabla^ih_i$ and $f$ satisfy the same differential equation $Lf=L\nabla^ih_i=0$; therefore also the difference $f-\nabla^ih_i$ satisfy the same differential equation $L(f-\nabla^ih_i)=0$.    
\end{proof}
Hence, both $f$ and $h$ satisfy a fourth order differential equation; the next step now is to exclude the values of $(k_1,k_2)$ which make the kernel of the differential operator $L$ non-trivial.
\begin{Lemma}[Kernel trivialization of $L$]\label{Lem1}
    Let $L:=\Delta^2+(\Bar{k}_2-k_1+4)\Delta-k_1\Bar{k}_2$ be a differential operator on $S^d$ and let $(k_1,k_2) \in (\mathbb{R}^+ \times \mathbb{R}^+)\setminus{\mathbb{J}}$,
where
\begin{equation}
    \mathbb{J}:=\bigg\{(k_1,k_2) \in (\mathbb{R}^+ \times \mathbb{R}^+)\  | \ \bigg(k_1>0,k_2=K(d-1)-K\mu_{\ell}+\frac{4K\mu_{\ell}}{k_1+K\mu_{\ell}}\bigg)\bigg\} 
\end{equation}
and $\mu_{\ell}:=\ell(\ell+d-1)$ with $\ell\geq0$. Then $\mathrm{Ker}(L)=\{0\}.$
\end{Lemma}
\begin{proof}
    Given
\begin{equation}
   L:= \Delta^2+(\Bar{k}_2-k_1+4)\Delta-k_1\Bar{k}_2
\end{equation}
let us study its kernel $\mathrm{Ker}(L)$. If $\lambda$ is an eigenvalue of the laplacian on the sphere, then an eigenfunction $Y$ such that
\begin{equation}
    \Delta Y = \lambda Y
\end{equation}
satisfies
\begin{equation}
    L Y = p(\lambda) Y, \qquad   p(\lambda) := \lambda^2 + (\bar{k}_2 - k_1 + 4)\lambda - k_1 \bar{k}_2.
\end{equation}
Therefore, the kernel is non-trivial if and only if there exists an eigenvalue $\lambda$ on the sphere such that $p(\lambda) = 0$, i.e
\begin{equation}
    \lambda^2 + (K(d-1) - k_2 - k_1 + 4)\lambda - k_1 (K(d-1) - k_2) = 0.
\end{equation}
Solving for $k_2$, we obtain
\begin{equation}
    k_2 = \frac{k_1 K(d-1) - \lambda^2 - (K(d-1) - k_1 + 4)\lambda}{k_1 - \lambda}=K(d-1)+\lambda-\frac{4\lambda}{k_1-\lambda}, \qquad k_1 \neq \lambda
\end{equation}
Thus, for each eigenvalue $\lambda$ on the sphere, this formula provides the curve of excluded values for $k_2$ as a function of $k_1$. For a $d$-dimensional sphere $S^d$ with sectional curvature $K>0$ we have $\lambda=-K\mu_{\ell}:=-K\ell(\ell+d-1),$
hence $\mathrm{Ker}(L)=0 \Leftrightarrow (k_1,k_2) \in (\mathbb{R}^+ \times \mathbb{R}^+)\setminus{\mathbb{J}}$.
\end{proof}
\begin{figure}[h]
    \centering
    {\includegraphics[width=.48\textwidth]{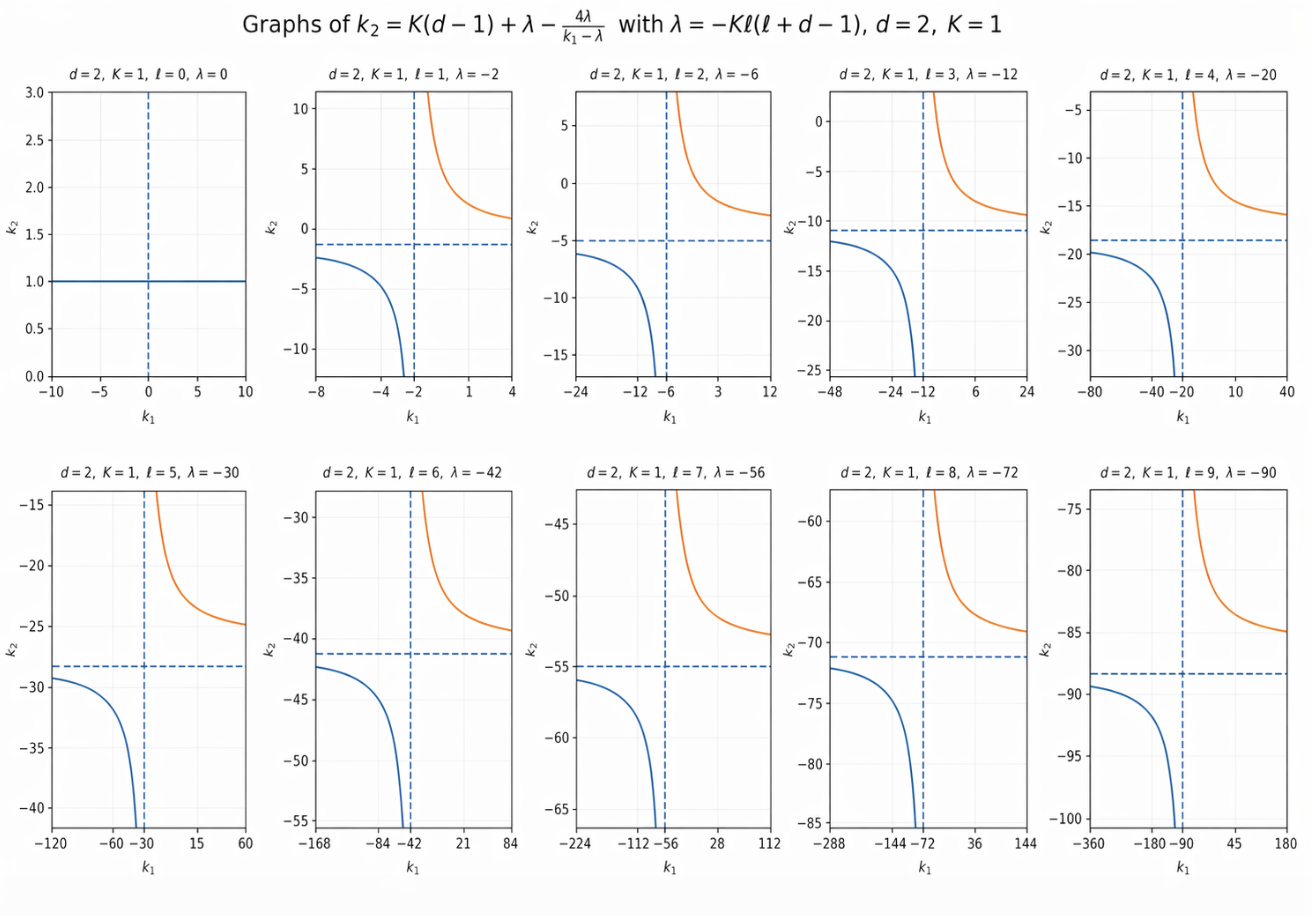}}
    \label{fig18a}\quad
    {\includegraphics[width=.48\textwidth]{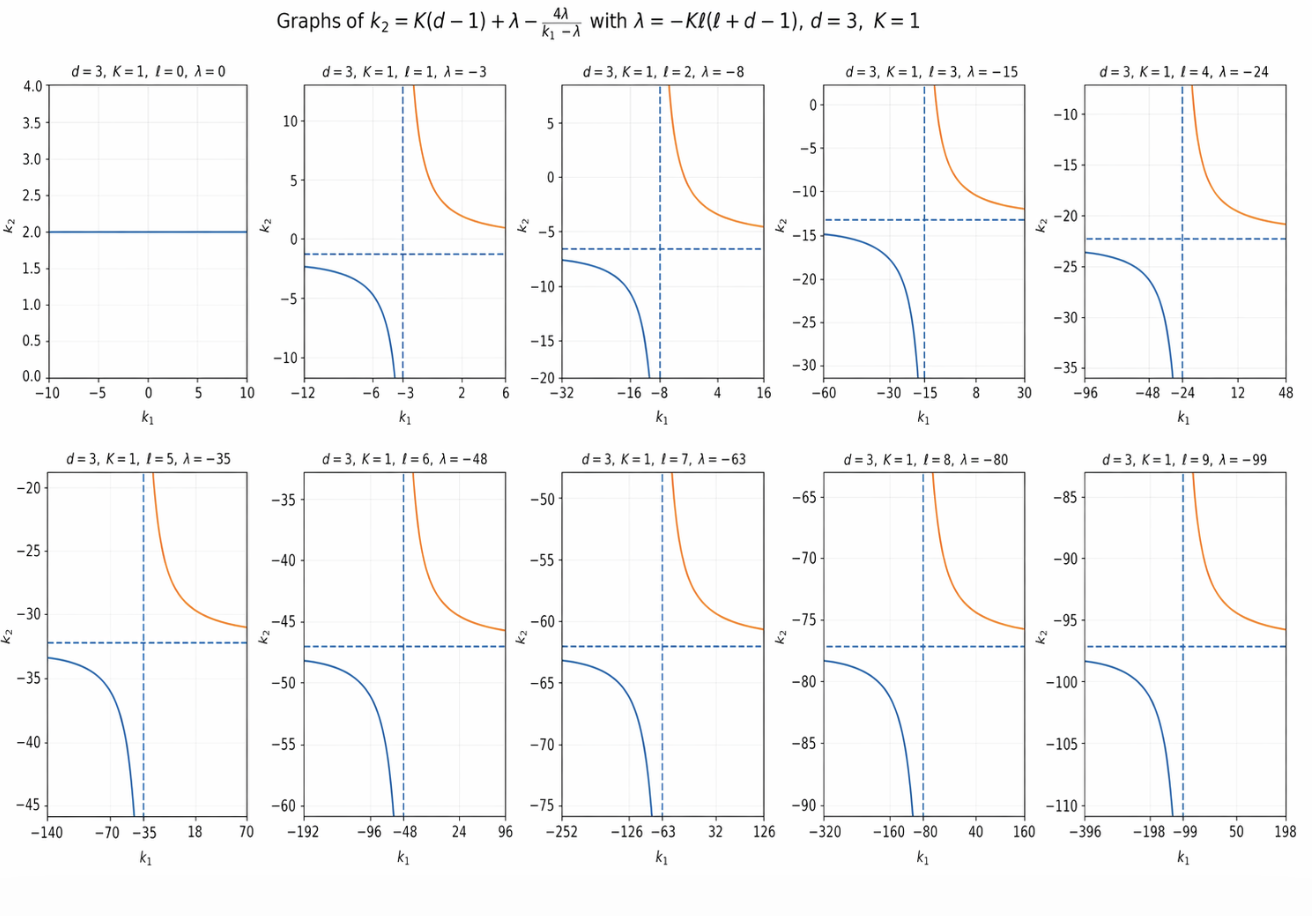}
    \label{fig18b}}
\caption{\textit{Graphs of curves that make the kernel of the $L$ operator non-trivial for $d=2,3$ and $K=1$.}}
\label{graph1}
\end{figure}
\begin{figure}[h]
    \centering
    {\includegraphics[width=.48\textwidth]{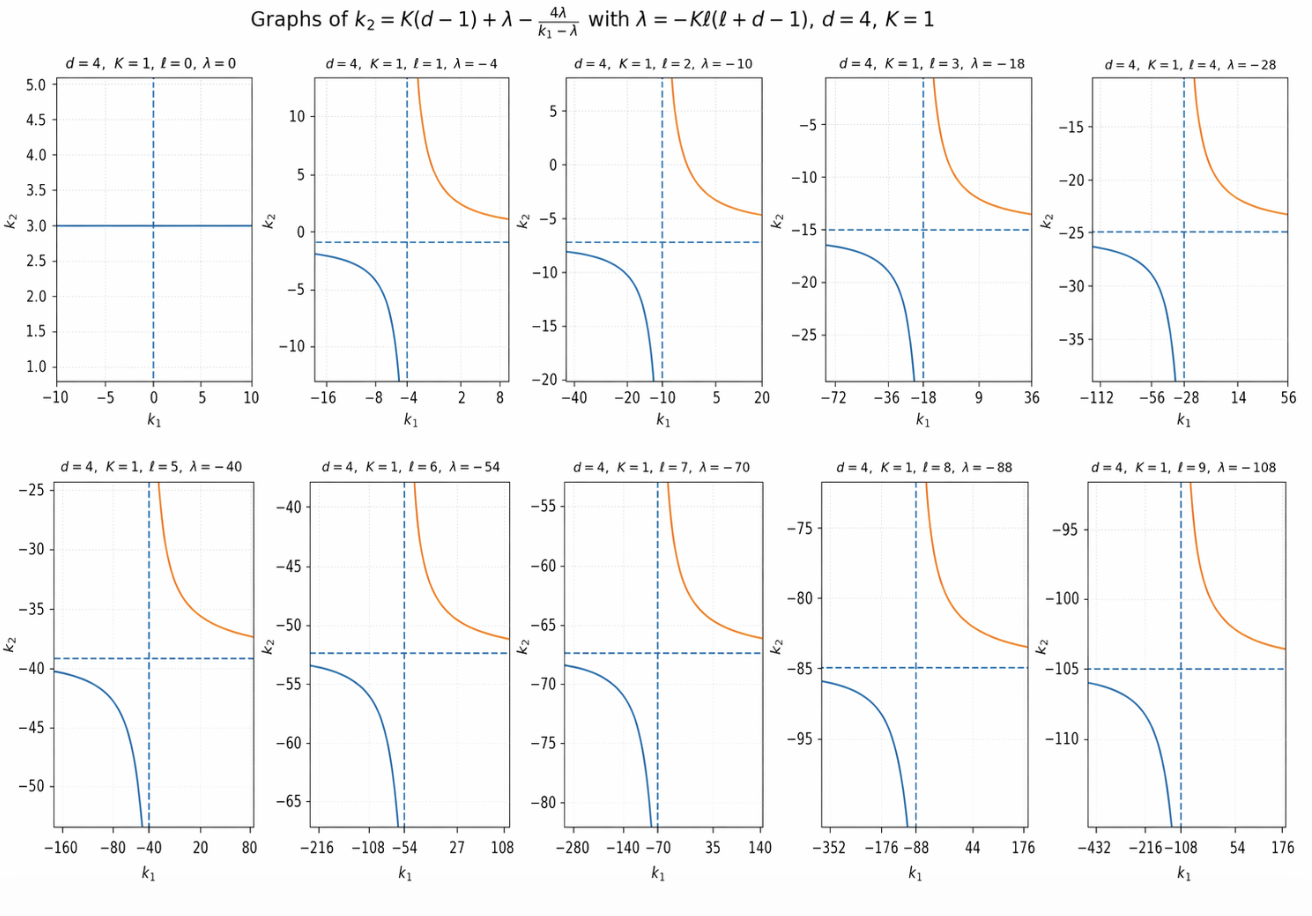}
    \label{fig18c}}
     {\includegraphics[width=.48\textwidth]{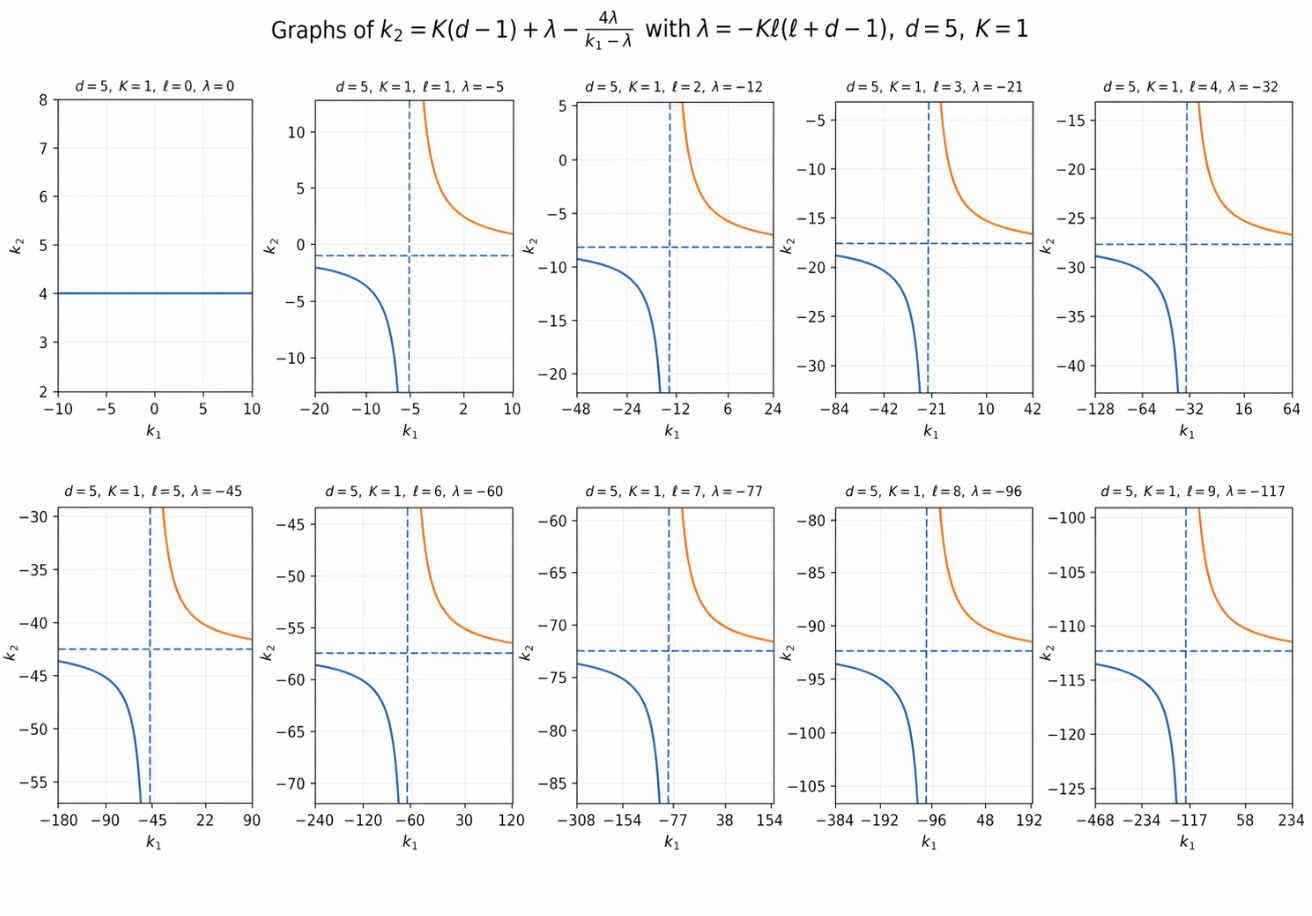}
    \label{fig18d}}
\caption{\textit{Graphs of curves that make the kernel of the $L$ operator non-trivial for $d=4,5$ and $K=1$.}}
\label{graph2}
\end{figure}
In Figure \ref{graph1} (above) and Figure \ref{graph2} (blow) are reported some graphs showing the curve of the plane $(k_1,k_2)$ that must be excluded for the kernel of the differential operator $L$ to be trivial. Recall the condition $k_1>0, k_2>0$.
\subsection{Trivialization theorems}
Let us now enounce the main theorem on the $d$-dimensional sphere.
\begin{Theorem}[Trivialization on $S^d$]\label{trivsys}
    Let $(S^d,\boldsymbol{g})$ be the $d$-dimensional sphere with sectional curvature $K>0$. Let $(k_1,k_2) \in (\mathbb{R}^+ \times \mathbb{R}^+)\setminus{\mathbb{J}}$. Then the harmonic-divgrad equations admit as only solution the trivial solution $f=0$ and $h=0$.
\end{Theorem}
\begin{proof}
By Proposition \ref{Prop1} we have that if $f$ and $h$ satisfy the harmonic-divgrad system then $L(f-\nabla^ih_i)=0$ and from Lemma \ref{Lem1} we have that if $(k_1,k_2) \in (\mathbb{R}^+ \times \mathbb{R}^+)\setminus{\mathbb{J}}$ then $\mathrm{Ker}(L)=0$, so we get
\begin{equation}
    f=\nabla^ih_i;
\end{equation}
using this information in equation \eqref{C1a} we can write
\begin{equation}
    [-k_1-2+\Delta]f=0 \quad \Rightarrow \quad f=0.
\end{equation}
since on $S^d$ the eigenvalues of $\Delta$ are non-positive. Once $f=0$ has been established\footnote{Note that this could be derived also from the fact that if $\mathrm{Ker}(L)=0$ and $Lf=0$ then $f=0$.}, equation \eqref{C1b} reduces to
\begin{equation}
\nabla^j \nabla_j h_i - k_2 h_i = 0 \, .
\label{eq:reduced-h}
\end{equation}
Since $h$ is a smooth $1$-form on the sphere $S^d$, we contract equation \eqref{eq:reduced-h} with $h^i$ and integrate over $S^d$ with respect to the Riemannian volume form $dV_{\boldsymbol{g}}$. This gives
\begin{equation}
\int_{S^d} h^i \nabla^j \nabla_j h_i \, dV_{\boldsymbol{g}}
-
k_2 \int_{S^d} h^i h_i \, dV_{\boldsymbol{g}}
=0 \, .
\label{eq:energy1}
\end{equation}
Using the metric, we can write $h^i h_i = |h|^2$, so that the second term is
\begin{equation}
\int_{S^d} h^i h_i \, dV_{\boldsymbol{g}}
=
\int_{S^d} |h|^2 \, dV_{\boldsymbol{g}} \, .
\label{eq:normh}
\end{equation}
It remains to compute the first term. Since $S^d$ is compact and without boundary, we can integrate by parts; consider the vector field
\begin{equation}
X^j := h^i \nabla^j h_i \, ,
\label{eq:vectorfield}
\end{equation}
its divergence is
\begin{equation}
\nabla_j X^j
=
\nabla_j \bigl( h^i \nabla^j h_i \bigr)
=
(\nabla_j h^i)(\nabla^j h_i) + h^i \nabla_j \nabla^j h_i \, .
\label{eq:divX}
\end{equation}
Integrating over $S^d$ and using the divergence theorem, the left-hand side vanishes:
\begin{equation}
\int_{S^d} \nabla_j X^j \, dV_{\boldsymbol{g}} = 0 \, .
\label{eq:divthm}
\end{equation}
Therefore,
\begin{equation}
\int_{S^d} h^i \nabla_j \nabla^j h_i \, dV_{\boldsymbol{g}}
=
-
\int_{S^d} (\nabla_j h^i)(\nabla^j h_i) \, dV_{\boldsymbol{g}} \, .
\label{eq:ibp1}
\end{equation}
Since the connection is metric-compatible, the quantity
\begin{equation}
(\nabla_j h^i)(\nabla^j h_i)
\end{equation}
is precisely the pointwise squared norm of $\nabla h$, namely
\begin{equation}
(\nabla_j h^i)(\nabla^j h_i) = |\nabla h|^2 \, .
\label{eq:normgrad}
\end{equation}
Hence
\begin{equation}
\int_{S^d} h^i \nabla^j \nabla_j h_i \, dV_{\boldsymbol{g}}
=
-
\int_{S^d} |\nabla h|^2 \, dV_{\boldsymbol{g}} \, .
\label{eq:ibp2}
\end{equation}

Substituting \eqref{eq:normh} and \eqref{eq:ibp2} into \eqref{eq:energy1}, we obtain
\begin{equation}
-
\int_{S^d} |\nabla h|^2 \, dV_{\boldsymbol{g}}
-
k_2 \int_{S^d} |h|^2 \, dV_{\boldsymbol{g}}
=0 \, ,
\end{equation}
or equivalently
\begin{equation}
\int_{S^d} |\nabla h|^2 \, dV_{\boldsymbol{g}}
+
k_2 \int_{S^d} |h|^2 \, dV_{\boldsymbol{g}}
=0 \, .
\label{eq:energyfinal}
\end{equation}
Both terms in \eqref{eq:energyfinal} are non-negative, and since $k_2>0$, the second term vanishes if and only if $h=0$. Therefore the identity \eqref{eq:energyfinal} implies
\begin{equation}
\int_{S^d} |\nabla h|^2 \, dV_g = 0 \, ,
\qquad
\int_{S^d} |h|^2 \, dV_g = 0 \, ,
\end{equation}
and hence
\begin{equation}
h=0 \, .
\label{eq:hzero}
\end{equation}
\end{proof}
At this point, using Killing-Hopf Theorem \ref{tkh} and the above Theorem \ref{trivsys} we can state the trivialization result on positive  sectional curvature space form.
\begin{Theorem}[Trivialization on a positive sectional curvature space form]\label{trivsysext}
    Let $(M,\boldsymbol{g})$ be a space form with sectional curvature $K>0$ of dimension $d$. Let $(k_1,k_2) \in (\mathbb{R}^+ \times \mathbb{R}^+)\setminus{\mathbb{J}}$. Then the harmonic-divgrad equations admit as only solution the trivial solution $f=0$ and $h=0$.
\end{Theorem}
\begin{proof}
    To prove the statement on a general positive sectional curvature space form, let
\[
\pi : S^d \to M
\]
be the Riemannian universal covering map. By the Killing--Hopf theorem, \(\pi\) is a local isometry. Let now \((f,h)\), with
\[
f \in C^\infty(M), \qquad h \in \Omega^1(M),
\]
be a solution of the harmonic-divgrad system on \(M\). We define on \(S^d\)
\[
\tilde f := f \circ \pi, \qquad \tilde h := \pi^\ast h.
\]
Since \(\pi\) is a local isometry, the Levi-Civita connection on \(S^d\) is the pullback of the Levi-Civita connection on \(M\) and therefore the natural differential operators appearing in the system commute with pullback. More precisely,
\begin{equation}
\widetilde{\Delta}(f \circ \pi)= (\Delta f)\circ \pi,
\qquad
d(f \circ \pi)=\pi^\ast(df),
\end{equation}
and, for every \(1\)-form \(h\),
\begin{equation}
\widetilde{\delta}(\pi^\ast h)= (\delta h)\circ \pi,
\qquad
\widetilde{\Delta}(\pi^\ast h)=\pi^\ast(\Delta h),
\end{equation}
where the operators with tildes are computed with respect to the metric on \(S^d\).\\
Therefore, pulling back the harmonic-divgrad system on \(M\),
\begin{subequations}
\begin{align}
(-k_1+\Delta)f-2\,\delta h&=0,\\
(-k_2+\Delta)h+2\,df&=0,
\end{align}
\end{subequations}
we obtain on \(S^d\)
\begin{subequations}
\begin{align}
(-k_1+\widetilde{\Delta})\tilde f-2\,\widetilde{\delta}\tilde h&=0,\\
(-k_2+\widetilde{\Delta})\tilde h+2\,d\tilde f&=0.
\end{align}
\end{subequations}
Hence \((\tilde f,\tilde h)\) is a solution of the same harmonic-divgrad system on the sphere \(S^d\).\\
Since \((k_1,k_2)\in (\mathbb{R}^+\times\mathbb{R}^+)\setminus\mathbb{J}\), Theorem \ref{trivsys} applies on \(S^d\) and yields
\begin{equation}
\tilde f=0,
\qquad
\tilde h=0
\qquad \quad
\text{on } S^d.
\end{equation}
That is,
\begin{equation}
f\circ\pi=0,
\qquad
\pi^\ast h=0.
\end{equation}
Since \(\pi\) is surjective, from \(f\circ\pi=0\) it follows immediately that \(f=0\) on \(M\). Indeed, for every \(x\in M\), choosing \(\tilde x\in S^d\) such that \(\pi(\tilde x)=x\), one has
\begin{equation}
f(x)=f(\pi(\tilde x))=(f\circ\pi)(\tilde x)=0.
\end{equation}

Similarly, \(\pi^\ast h=0\) implies \(h=0\). In fact, since \(\pi\) is a local diffeomorphism, for every \(\tilde x\in S^d\) the differential
\[
d\pi_{\tilde x}:T_{\tilde x}S^d\to T_{\pi(\tilde x)}M
\]
is an isomorphism. If \(h_{\pi(\tilde x)}\neq 0\), then its pullback through the isomorphism \(d\pi_{\tilde x}\) could not vanish, contradicting \((\pi^\ast h)_{\tilde x}=0\). Therefore \(h=0\) on \(M\).
\end{proof}
\section{Conclusions and outlook}
In this work we analyzed the harmonic-divgrad system on Riemannian manifolds of constant sectional curvature, Definition \ref{def}, and proved, in Theorem \ref{trivsysext}, that away from the set $\mathbb{J}$, the system admits only the trivial solution on positive sectional curvature space forms. The result is useful for two related reasons. First, it provides a sharp obstruction to the existence of non-trivial smooth modes of the coupled scalar–one-form system, thereby could considerably simplify the analysis of gauge parameters and asymptotic data in higher-form and mixed-symmetry gauge theories. Second, it shows that the trivialization mechanism is not tied to the round sphere alone, but extends to all positive-curvature space forms through the Killing-Hopf theorem and therefore to non-simply connected spherical quotients as well. This widens the geometric scope of the result in a natural and physically meaningful way.

From the point of view of asymptotic symmetries, this suggests that whenever the geometry at null-infinity is modeled not by a sphere but by a positive sectional curvature space form, the same rigidity phenomenon persists. On the one hand, this means that any non-vanishing asymptotic charge or residual gauge parameter must originate from special loci in parameter space or from weaker boundary conditions,  polyhomogeneous sectors, additional topological or singular structures not captured by the analyses which lead to the emergence of the harmonic-divgrad system in gauge theories. This makes the present Theorem \ref{trivsysext} a useful diagnostic tool in the classification of admissible asymptotic sectors. On the other hand, Theorem \ref{trivsysext} can be used to understand when a gauge fixing or a gauge-for-gauge fixing is possible, like in the case of $p$-form gauge theory.

The extension to lens spaces is especially interesting. Since lens spaces are spherical space forms, the trivialization theorem applies directly to them. This opens the possibility of studying higher-form and mixed-symmetry gauge theories on backgrounds where the angular geometry is a lens-space quotient rather than a round sphere, a setting that naturally appears in string-theoretic and holographic constructions. In those contexts, replacing the sphere by a lens space changes the global topology while preserving positive constant sectional curvature locally, so the present result isolates precisely what remains rigid and what may instead depend on the global quotient structure. This may be relevant, for instance, in the analysis of gauge sectors, asymptotic charges and boundary data in quiver gauge theories emerging from holographic models built on orbifold or quotient geometries.

Several directions deserve further investigation. A first one is to understand if we can drop the hypothesis of $k_1>0,k_2>0$ and modify appropriately the set $\mathbb{J}$. A second is to generalize the present analysis to systems involving higher-degree forms or mixed-symmetry tensors, where analogous coupled elliptic operators arise in gauge-for-gauge chains and in the study of asymptotic symmetries. A third is to explore whether, on non-simply connected positive-curvature space forms such as lens spaces, the interplay between local trivialization and global topology can produce physically meaningful sectors once one relaxes smoothness assumptions or includes distributional, twisted or cohomologically non-trivial data. Finally the extension to space form with vanishing or negative sectional curvature. In these cases the arguments
developed above cannot be applied directly on the universal covers
$\mathbb{E}^d$ or $\mathbb{H}^d$, since these spaces are non-compact.
One should instead work directly on the compact  and without boundary quotient of interest $M$. Typical examples include flat tori and Klein bottles, which arise in toroidal
compactifications, finite-volume quantum field theory and non-orientable
two-dimensional models, as well as closed hyperbolic surfaces and compact
hyperbolic three-manifolds, which appear respectively as higher-genus string
worldsheets and as spatial sections in cosmological models with negative
curvature and non-trivial topology. 
The reduction to the fourth-order scalar operator is expected to remain valid,
as it only relies on the identity between the Riemann tensor and the constant sectional curvature. Moreover, compactness allows one to use the energetic argument directly on the quotient. However, the corresponding
resonant set should then be defined by
\[
\mathbb{J}_M:=\left\{(k_1,k_2)\in\mathbb{R}_+\times \mathbb{R}_+\ \middle|\ 
p(\lambda)=0
\text{ for some }\lambda\in\operatorname{Spec}(\Delta_M)\right\}.
\]
Thus, outside $\mathbb{J}_M$, one expects the same trivialization result, although
$\mathbb{J}_M$ now depends on the spectrum of the specific compact quotient.

These problems would help clarify the role of geometry and topology in asymptotic gauge dynamics beyond the standard spherical setting.

\section*{Conflict of Interest, Funding and Data Availability}
The author declares no conflict of interest. This research received no external funding. No new data were created or analyzed in this study.

\bibliographystyle{JHEP} 
\bibliography{biblio}

\end{document}